\documentclass[a4paper]{article}
\usepackage[top=25truemm,bottom=20truemm,left=25truemm,right=25truemm]{geometry}\usepackage{CJK}
\usepackage[CJKbookmarks]{hyperref}
\usepackage[dvipdfmx]{graphicx,xcolor}
\usepackage[T1]{fontenc}
\usepackage{units}
\usepackage{lmodern}
\usepackage{textcomp}
\usepackage[leqno]{amsmath}
\usepackage{amssymb}
\usepackage{amsthm}
\usepackage{url}
\usepackage{natbib}
\makeatletter
\providecommand*{\shuffle}{%
  \mathbin{\mathpalette\shuffle@{}}%
}
\newcommand*{\shuffle@}[2]{%
  \sbox0{$#1\vcenter{}$}%
  \kern .15\ht0 
  \rlap{\vrule height .25\ht0 depth 0pt width 2.5\ht0}%
  \raise.1\ht0\hbox to 2.5\ht0{%
    \vrule height 1.75\ht0 depth -.1\ht0 width .17\ht0 %
    \hfill
    \vrule height 1.75\ht0 depth -.1\ht0 width .17\ht0 %
    \hfill
    \vrule height 1.75\ht0 depth -.1\ht0 width .17\ht0 %
  }%
  \kern .15\ht0 
}
\makeatother

\def\AmSLaTeX{\leavevmode\hbox{$\mathcal{A}\kern-.2em\lower.376ex
 \hbox{$\mathcal{M}$}\kern-.2em\mathcal{S}$-\LaTeX}}
\def\BibTeX{{B\kern-.05em{\textsc{i}\kern-.025em\textsc{b}}\kern-.08em
 \TeX}}
\newtheorem{theorem}{定理\label{universal}}

\begin{document}
\begin{CJK}{UTF8}{min}
  \title{Introduction to the signature method\\
    シグネチャ法入門}
\author{Nozomi Sugiura}


\begin{abstract}
    The sequential data observed in earth science can be regarded as paths in multidimensional space.
  To read the path effectively, it is useful to convert it
  into a sequence of numbers called the signature,
  which can faithfully describe the order of
  points and nonlinearity in the path.
  In particular, a linear combination of the terms in a signature can be used to approximate any nonlinear function defined on a set of paths.
  Thereby, when one 
  learns a set of sequential data with labels attached to it,
  linear regression can be applied to the pairs of
  signature and label, which will 
  achieve high performance learning even when the labels
  are determined by a nonlinear function.
  By incorporating the signature methods
  into machine learning and data assimilation
  utilizing sequential data, it is expected
  that we can extract information that has previously been overlooked.
\end{abstract}

\maketitle


\section{はじめに}
地球科学をはじめとする様々な実証的な科学研究においては，
系列データを分析することがしばしば必要になる．
ここで，系列データというのは，例えばあるパラメータに沿って値が定義された
多次元空間上の経路$X: [0,1] \ni t \mapsto X_t \in \mathbb{R}^d$
の上のいくつかの点$X_{t_1},\cdots,X_{t_m}$を観測したもののことである．
これらをばらばらの点と捉え，モデルと比較したり機械学習したりするのが素朴な扱い方であろう．
しかし，これらの点が順番に並んでいることに意味があることもありえる．
本稿では，順番の情報を失わないように，系列データを効率的に読み取る方法について述べる．
この方法においては，
経路上の点群として与えられる観測データを，
シグネチャ\citep[e.g.,][]{lyons2007differential,frizvictoir2010}と呼ばれる数列に
変換することが核心にあるため，シグネチャ法と呼ぶことにする．
シグネチャは，ラフパス理論\citep{lyons1998differential}という比較的
新しい数学理論における主要概念のひとつである．

系列データをシグネチャに変換することにより，
高性能かつ効率的な機械学習が可能になることが過去の研究において示されている．
例えば，\citet{FERMANIAN2021107148}には，
公開データMotion Sense \citep{malekzadeh2018protecting},
Urban Sound \citep{salamon2014dataset}を
対象とした機械学習において，
シグネチャ法が他の最先端手法に比して優位かまたは同程度の性能を
効率的に達成することが示されている．
一方，Quick Draw! \citep{quick}に関しては，極めて高い計算効率で
一定の性能が得られている．
\citet{10.1609/aaai.v33i01.33018585}には，
人間の動作把握への応用例において，
シグネチャ法が他手法よりも好成績を効率的に達成することが示されている．
その他にも，
医療データ \citep{WOS:000454240100001,9005805,WOS:000532262800030}，
文字認識 \citep{WOS:000437271100009}，
金融時系列 \citep{10.1145/2640087.2644157}，
地球科学 \citep{Sugiura2020}など多方面への応用例がある．

本稿では，シグネチャを使う理由を述べた後，シグネチャの定義と性質を示し，
最後に機械学習への応用について議論する．
説明にあたっては，原理的な側面が主になるものの，
数学的な厳密性よりも応用とのつながりを重視する．
\section{動機づけ}
系列データを扱う際になぜシグネチャが必要になるかを２つの例を挙げて説明する．
\subsection{多項式回帰\label{polynomial}}
点とそこでの評価値の集合$\{(x_i,y_i)|
~x_i,y_i \in \mathbb{R},~i=1,\cdots,N\}$
を考える．
このデータを多項式回帰するには，
次のコスト関数を最小化することで，多項式の係数$w=(w_0,\cdots,w_M)$を最適化すればいい
(図\ref{fig:reg})．
\begin{align}
  \label{reg0}
  J(w) &= \sum_{i=1}^{N}\left(
  y_i - \sum_{k=0}^M w_k x_i^k \right)^2. 
\end{align}
これと同様のことを経路の集合に対して行いたい．
経路とそれに対する評価値の集合$\{(X_i,y_i)|~y_i \in \mathbb{R},
~i=1,\cdots,N\}$が与えられているとする．
各経路$X_i$は区間$[0,1]$から$\mathbb{R}^d$への写像$X_i: t\mapsto X_{i,t}$である．
各経路をシグネチャに変換することにより，経路の集合に対する回帰が可能になる．
具体的には，次のコスト関数を最小化すればいい．
\begin{align}
  \label{reg1}
  J(w) &= \sum_{i=1}^{N}\left(
  y_i - \sum_{k=0}^M w_k \mathcal{S}^{(k)}(X_i) \right)^2. 
\end{align}
ここで$\mathcal{S}^{(k)}(X_i)$は経路$X_i$に対して定められたシグネチャと
呼ばれる数列の第$k$番目の項である．
ただし，以下ではシグネチャの各項$(k)$を多重添字で付番することがある．
このように各経路をあたかもひとつの点のように扱って「多項式」回帰することを可能にするのが，シグネチャ法の役目のひとつである．

点の集合に対する多項式回帰においては，単項式$1,x,\cdots,x^m$が回帰曲線(図\ref{fig:reg})の
基底関数であるのに対して，
経路の集合に対する回帰においては，シグネチャ$\mathcal{S}^{(0)}(X),\cdots,\mathcal{S}^{(M)}(X)$が回帰超曲面の基底関数となる．
つまり，図\ref{fig:reg}の横軸を，点の集合ではなく，
経路の集合に置き換えた場合に相当する．
\subsection{非可換テイラー展開\label{taylor}}
区間$[0,t]$から$d$次元ユークリッド空間$\mathbb{R}^d$への写像$X: u \mapsto X_u$を経路と呼ぶ．
$Y: [0,t]\to \mathbb{R}^e$は，$d$個の滑らかな
ベクトル場$V_i:  \mathbb{R}^e \to \mathbb{R}^e, ~i=1,\cdots,d$を用いた次の常微分方程式で定められるとする．
\begin{align}
\begin{bmatrix}
  dY^{(1)}_u\\
  \vdots\\
  dY^{(e)}_u
\end{bmatrix}
&=
\begin{bmatrix}
V_{1}^{(1)}(Y_u)&\cdots &V_{d}^{(1)}(Y_u)\\
\vdots    &       &\vdots    \\
V_{1}^{(e)}(Y_u)&\cdots &,V_{d}^{(e)}(Y_u)
\end{bmatrix}
\begin{bmatrix}
  dX^{(1)}_u\\
  \vdots\\
  dX^{(d)}_u
\end{bmatrix}.
  \label{ode1}
\end{align}
これは$dY_u = V(Y_u)dX_u$とも書ける．
常微分方程式(\ref{ode1})には，
再帰型ニューラルネット\citep{liao2019learning}や確率微分方程式など
多くの例が含まれる．
$Y$は$X$に駆動される制御型方程式になっている．

$Y_t$の滑らかな関数$F: \mathbb{R}^e \to \mathbb{R}$による評価値
$F(Y_t)$を$X$に沿ってテイラー展開する
\citep[e.g.,][]{Litterer2011OnAC,baudoin2012taylor}．テイラー展開は，微分積分学の基本定理の適用を繰り返すことによって導出することができる．
なお，以下で上付きと下付きの同一の添字の組に対しては，総和を施すものとする．

関数$F$の方向微分を
$\nabla_{V_i}F(Y_u):=V^{(k)}_i(Y_u)\frac{\partial}{\partial Y^{(k)}}F(Y_u)$
と書くと，
\begin{align}
dF_{u_1} &= \frac{\partial F}{\partial Y^{(k_1)}}(Y_{u_1})dY^{(k_1)}_{u_1}
=\nabla_{V_{i_1}}F(Y_{u_1})dX^{(i_1)}_{u_1}
\end{align}
なので， 微分積分学の基本定理より，
\begin{align}
  F(Y_t)
  &=
  F(Y_0)+
  \int_{0<u_1<t} \nabla_{V_{i_1}} F(Y_{u_1})
  dX^{(i_1)}_{u_1}.
  \label{taylor21a}
\end{align}
式(\ref{taylor21a})内の$\nabla_{V_{i_1}} F(Y_{u_1})$に微分積分学の基本定理を適用
して代入すると，
\begin{align}
  F(Y_t)
  &=
  F(Y_0)
  +
  \nabla_{V_{i_1}} F(Y_0)  \int_{0<u_1<t}  dX^{(i_1)}_{u_1}
  +
  \int_{0<u_1<u_2<t}
    \nabla_{V_{i_1}}\nabla_{V_{i_2}} F(Y_{u_1})  
    dX^{(i_1)}_{u_1}  dX^{(i_2)}_{u_2}.\label{taylor21c}
\end{align}
この操作を繰り返し適用することにより，
以下のように
$n$次の非可換テイラー展開が得られる\citep[e.g.,][]{10.1214/ECP.v20-4124}．
\begin{align}
  F(Y_t)
  &=  \sum_{k=0}^{n}
  \nabla_{V_{i_1}} \cdots \nabla_{V_{i_k}}  F(Y_0)
  \int_{0<u_1<\cdots<u_k<t}
  dX^{(i_1)}_{u_1} \cdots dX^{(i_k)}_{u_k}+R_{n+1}(t),
    \label{taylor22}\\
    R_{n+1}(t)
    &=
  \int_{0<u_1<\cdots<u_n<t}
  \left(    \nabla_{V_{i_1}} \cdots \nabla_{V_{i_n}}  F(Y_{u_1})-
    \nabla_{V_{i_1}} \cdots \nabla_{V_{i_n}}  F(Y_0)    \right)
  dX^{(i_1)}_{u_1} \cdots dX^{(i_n)}_{u_n}.
  \label{taylor22a}
\end{align}
式(\ref{taylor22})の２次以上の項においては，
方向微分が可換でないため，
通常のテイラー展開における$X_t^{(i)}-X_0^{(i)}$の冪のかわりに，
反復積分（後述）
$\int_{0<u_1<\cdots<u_k<t}
  dX^{(i_1)}_{u_1} \cdots dX^{(i_k)}_{u_k}$
が現れる．

式(\ref{taylor22})が示唆するのは，
システム$\{Y_0,F,V_1,\cdots,V_d\}$
が未知な場合，
  評価値$F(Y_t)$と
  経路の反復積分
  $\{\{ \int_{0<u_1<\cdots<u_k<t}dX^{(i_1)}_{u_1} \cdots dX^{(i_k)}_{u_k}
  \}_{i_1,\cdots i_k=1,\cdots,d} \}_{k=1,\cdots,n}$
  との組をデータサンプルとして持っていれば，
  重回帰によりシステム同定ができるということである．
  

\section{シグネチャの概要}
経路に対して定まるシグネチャという数列の定義と性質について述べる．
特に，シグネチャが経路を独立変数とする関数の空間における基底関数になっていることが重要である．
\subsection{シグネチャの定義}
$X: [0,1]\to \mathbb{R}^d$を経路とする．
また，区間$[0,1]$の分割を$D$として，経路の長さを
$\ell=\sup_{D}\sum_{[t_{i-1},t_i]\in D}|X_{t_i}-X_{t_{i-1}}|$で定義する．
この長さ$\ell$が有限であることを有界変動を持つという．
有界変動を持つ経路に対して，$1$次の反復積分を次のように定義する．$i_1=1,\cdots,d$に対して，
\begin{align}
  \mathcal{S}^{(i_1)}(X)
  &= \int_{0<t_1<1} dX^{(i_1)}_{t_1}=  X^{(i_1)}_1-X^{(i_1)}_0.
\end{align}
$2$次の反復積分は，$i_1,i_2=1,\cdots,d$に対して，
\begin{align}
  \mathcal{S}^{(i_1i_2)}(X) &=
  \int_{0<t_2<1} \mathcal{S}^{(i_1)}(X_{[0,t_2]}) dX^{(i_2)}_{t_2}
  =
  \int_{0<t_1<t_2<1} dX^{(i_1)}_{t_1} dX^{(i_2)}_{t_2}.
\end{align}
これを続けていくと，$n$次の反復積分は，$i_1,\cdots,i_n=1,\cdots,d$に対して，
\begin{align}
  \label{sig_n}
  \mathcal{S}^{(i_1\cdots i_n)}(X) &=
  \int_{0<t_n<1} \mathcal{S}^{(i_1\cdots i_{n-1})}(X_{[0,t_n]}) dX^{(i_n)}_{t_n}
  =
  \int_{0<t_1<\cdots < t_{n-1}<t_n<1}
  dX^{(i_1)}_{t_1}    \cdots  dX^{(i_{n-1})}_{t_{n-1}}dX^{(i_n)}_{t_n}.
\end{align}
$n$段シグネチャ$\mathcal{S}_n(X)$は，これらを並べて書いたものである．
\begin{align}
  \label{signature}
  \mathcal{S}_n(X) &= \left(  \mathcal{S}^{()}(X), 
  \mathcal{S}^{(\bullet)}(X),
  \mathcal{S}^{(\bullet \bullet)}(X),
  \cdots,
  \mathcal{S}^{(\overbrace{\bullet\cdots \bullet}^{n\text{個}})}(X)
  \right).
\end{align}
ここで，$\bullet$には，それぞれ$1,\cdots,d$が入る．
$0$次の反復積分$\mathcal{S}^{()}(X)$は，常に$1$と定義する．
なお，式(\ref{signature})で$n\to \infty$とした無限列のことを
$\mathcal{S}(X)$と書き，無限段のシグネチャ，または単にシグネチャと呼ぶ．

具体例として，$d=2$の場合の$2$段シグネチャを書くと，
\begin{align}
  \label{sig2}
  \mathcal{S}_2(X)
  &=
 \left(\mathcal{S}^{()}(X),
\begin{bmatrix}
  \mathcal{S}^{(1)}(X)\\
  \mathcal{S}^{(2)}(X)
\end{bmatrix},
  \begin{bmatrix}
  \mathcal{S}^{(11)}(X)&  \mathcal{S}^{(12)}(X)\\
  \mathcal{S}^{(21)}(X)&  \mathcal{S}^{(22)}(X)
\end{bmatrix}
  \right)\\
&=
\left(1,
  \begin{bmatrix}
    X^{(1)}_1-X^{(1)}_0\\
    X^{(2)}_1-X^{(2)}_0
\end{bmatrix},
  \begin{bmatrix}
\frac12 (X^{(1)}_{1}-X^{(1)}_0)^2&
\int_0^1(X^{(1)}_{t_2}-X^{(1)}_0)dX^{(2)}_{t_2}\\
\int_0^1(X^{(2)}_{t_2}-X^{(2)}_0)dX^{(1)}_{t_2}&
\frac12 (X^{(2)}_{1}-X^{(2)}_0)^2&
\end{bmatrix}
  \right).  \nonumber
\end{align}
式(\ref{sig2})の$2$次の反復積分のうち，
２つの非対角項の差（の２分の１）は
$1$次の反復積分$X^{(i)}_1-X^{(i)}_0$の冪では表わせないことに注意(L\'{e}vy面積とも呼ばれる)．
$n$段シグネチャの成分は，$1$も含めると，$1+d+d^2+\cdots+d^n=(d^{n+1}-1)/(d-1)$個ある．

より正確には，付録\ref{tensor}の式(\ref{sig_tensor})に示すように，
シグネチャはテンソル代数$T^n(\mathbb{R}^d)$の元とみなすことができる．

\subsection{経路とシグネチャとの対応}
区間$[0,1]$で定義された経路$X$を時刻$s$を境に２分割することを考える．
この節では，パラメータの範囲がわかるようにこの経路を$X_{[0,1]}$というようにも書く．
$2$階の反復積分の
第$(i_1,i_2)$成分
\begin{align}
\mathcal{S}^{(i_1i_2)}(X_{[0,1]})
&=\int_{0<t_1<t_2<1} dX_{t_1}^{(i_1)} dX_{t_2}^{(i_2)}
\end{align}
を考える．
$t_1$と$t_2$に関する積分範囲は$0<t_1<t_2<1$であるが，
任意の$0<s <1$を採ると，この範囲は，
\begin{align}
\{s<t_1<t_2<1\} \cup \{0<t_1<s<t_2<1\} \cup \{0<t_1<t_2<s\}&
\end{align}
と３つの互いに交わらない部分集合に分けることができる．これにより，
\begin{align}
\mathcal{S}^{(i_1i_2)}(X_{[0,1]})
&=\int_{s<t_1<t_2<1} dX_{t_1}^{(i_1)} dX_{t_2}^{(i_2)}
+\int_{0< t_1 < s < t_2 <1} dX_{t_1}^{(i_1)} dX_{t_2}^{(i_2)}
+\int_{0<t_1<t_2<s} dX_{t_1}^{(i_1)} dX_{t_2}^{(i_2)} \\
&=
\int_{s<t_1<t_2<1} dX_{t_1}^{(i_1)} dX_{t_2}^{(i_2)}+\int_{0< t_1 < s} dX_{t_1}^{(i_1)}
\int_{s<t_2<1}dX_{t_2}^{(i_2)}
+\int_{0<t_1<t_2<s} dX_{t_1}^{(i_1)} dX_{t_2}^{(i_2)}\nonumber\\
&=
\mathcal{S}^{(i_1i_2)}(X_{[s,1]})
+\mathcal{S}^{(i_1)}(X_{[0,s]})\mathcal{S}^{(i_2)}(X_{[s,1]})
+\mathcal{S}^{(i_1i_2)}(X_{[0,s]}).\nonumber
\end{align}
この式が示すのは，経路$X_{[0,s]}$と$X_{[s,1]}$とを繋げたときの反復積分の計算規則である．

１階の反復積分の計算に関しては，上記の第２項を考えなくていいので，
\begin{align}
\mathcal{S}^{(i_1)}(X_{[0,1]})&=
\mathcal{S}^{(i_1)}(X_{[0,s]})+\mathcal{S}^{(i_1)}(X_{[s,1]})
\end{align}
このような計算規則は，容易に高階の反復積分に拡張できる：
\begin{align}
  \label{chen0}
  \mathcal{S}^{(i_1\cdots i_n)}(X_{[0,1]})
  &=
  \sum_{k=0}^{n} \mathcal{S}^{(i_1\cdots i_k)}(X_{[0,s]})
  \mathcal{S}^{(i_{k+1}\cdots i_n)}(X_{[s,1]}).
\end{align}
これら反復積分の計算規則は，付録\ref{tensor}の
式(\ref{tensor_prod})のテンソル積と整合的であることがわかる．
すなわち，経路$X_{[0,s]}$と$X_{[s,1]}$とを繋げた経路を$X_{[0,s]}*X_{[s,1]}$と書くと，
\begin{align}
  \label{chen}
  \mathcal{S}(X_{[0,1]})&=  \mathcal{S}(X_{[0,s]}*X_{[s,1]}  )=
  \mathcal{S}(X_{[0,s]})\otimes \mathcal{S}(X_{[s,1]})
\end{align}
という式が成り立つ．
これがChenの恒等式であり，経路を繋げるという操作($*$)とシグネチャのテンソル積($\otimes$)
の間の関係（準同型）を表している．

さらに，シグネチャが経路を忠実に表現していることを示す次の事実がある
\citep{hambly2010uniqueness}．
２つの経路$X,Y$は，同じ経路を往復で辿るような枝状の部分(樹状経路)
の違いを除いて等しいとき，樹状同値であるといい，$X\sim Y$と書く．
有界変動を持つ２つの経路は，樹状同値であるとき，またそのときに限り，
（無限段の）シグネチャが等しい．すなわち，
\begin{align}
    X \sim Y&\iff \mathcal{S}(X)=\mathcal{S}(Y).
  \label{unique}
\end{align}

\subsection{シグネチャの成分どうしの積}
経路$X$を固定し，そのシグネチャの成分（反復積分）の間に成り立つ関係式を導出する．
$d$次元経路$X$のシグネチャの第$(i_1,i_2)$成分は，
\begin{align}
  \mathcal{S}^{(i_1i_2)}(X)&=\int_{0<t_1<t_2<1}dX_{t_1}^{(i_1)}dX_{t_2}^{(i_2)}
\end{align}
であるが，これと経路のシグネチャの第$(i_3,i_4)$成分との（実数どうしの）積は，
\begin{align}
  \mathcal{S}^{(i_1i_2)}(X)  \mathcal{S}^{(i_3i_4)}(X)
  &=
\left(  \int_{0<t_1<t_2<1}dX_{t_1}^{(i_1)}dX_{t_2}^{(i_2)} \right)
\left(  \int_{0<t_3<t_4<1}dX_{t_3}^{(i_3)}dX_{t_4}^{(i_4)} \right)
\end{align}
であるが，重積分の形で書くと，
\begin{align}
  \mathcal{S}^{(i_1i_2)}(X)  \mathcal{S}^{(i_3i_4)}(X)
  &=
  \int_{0<t_1<t_2<t_3<t_4<1}dX_{t_1}^{(i_1)}dX_{t_2}^{(i_2)}dX_{t_3}^{(i_3)}dX_{t_4}^{(i_4)}
  +\\
  &\cdots+
  \int_{0<t_3<t_4<t_1<t_2<1}dX_{t_1}^{(i_1)}dX_{t_2}^{(i_2)}dX_{t_3}^{(i_3)}dX_{t_4}^{(i_4)}.
  \nonumber
\end{align}
ここで，積分範囲は$t_1<t_2$および$t_3<t_4$の順序を保ちつつ，$t_1,t_2,t_3,t_4$を
並べ替えるすべての順列に亘る．
このような順列のとり方をシャッフル積といい，この場合は，
\begin{align}
12\shuffle 34 &= \{1234, 1324, 1342, 3124, 3142, 3412\}
\end{align}
となる．これを用いて積を書くと，
\begin{align}
  \mathcal{S}^{(i_1i_2)}(X)  \mathcal{S}^{(i_3i_4)}(X)
  &=
  \sum_{\sigma \in 12\shuffle 34}
  \int_{0<t_{\sigma(1)}<t_{\sigma(2)}<t_{\sigma(3)}<t_{\sigma(4)}<1}
  dX_{t_1}^{(i_1)}dX_{t_2}^{(i_2)}dX_{t_3}^{(i_3)}dX_{t_4}^{(i_4)}\\
  &=
  \sum_{\sigma \in 12\shuffle 34}
  \int_{0<t_{1}<t_{2}<t_{3}<t_{4}<1}
    dX_{t_{\sigma^{-1}(1)}}^{(i_1)}dX_{t_{\sigma^{-1}(2)}}^{(i_2)}
    dX_{t_{\sigma^{-1}(3)}}^{(i_3)}dX_{t_{\sigma^{-1}(4)}}^{(i_4)}\nonumber\\
  &=
  \sum_{\sigma \in 12\shuffle 34}
  \int_{0<t_{1}<t_{2}<t_{3}<t_{4}<1}
  dX_{t_1}^{(i_{\sigma(1)})}
  dX_{t_2}^{(i_{\sigma(2)})}
  dX_{t_3}^{(i_{\sigma(3)})}
  dX_{t_4}^{(i_{\sigma(4)})}\nonumber\\
  &=
  \sum_{\sigma \in 12\shuffle34}
  \mathcal{S}^{(i_{\sigma(1)}i_{\sigma(2)}i_{\sigma(3)}i_{\sigma(4)})}(X).\nonumber
\end{align}
最後の式は，$\mathcal{S}^{(i_1i_2\shuffle i_3i_4)}(X)$と略記される．
すなわち，
\begin{align}
\label{shuffle2a}
  \mathcal{S}^{(i_1i_2)}(X)  \mathcal{S}^{(i_3i_4)}(X)
  &=\mathcal{S}^{(i_1i_2\shuffle i_3i_4)}(X).
  \end{align}
式(\ref{shuffle2a})は，任意の多重添字$I,J$に拡張できて，
\begin{align}
\label{shuffle}
  \mathcal{S}^{(I)}(X)  \mathcal{S}^{(J)}(X)
  &=\mathcal{S}^{(I\shuffle J)}(X).
\end{align}
つまり，２つの反復積分の積は，より高次の反復積分の和で書き表せることがわかる．

さらに，経路$X$に対する
反復積分の線形結合$f(X)=\sum_I v_I\mathcal{S}^{(I)}(X)$と
$g(X)=\sum_J w_J\mathcal{S}^{(J)}(X)$との積は，
式(\ref{shuffle})より，
\begin{align}
f(X) g(X)&=\sum_{I,J} v_I w_J\mathcal{S}^{(I\shuffle J)}(X)
\label{shuffle_fg}
\end{align}
となるから，やはり反復積分の線形結合になっている．
すなわち，反復積分の線形結合は積に関して閉じており，
このことが以下の近似定理につながる．
\subsection{経路の関数に対する近似定理\label{approx_theory}}
有界変動を持つ経路の集合上で定義された
連続関数は，常に反復積分の線形結合で一様に近似できる．
このことを，経路の関数に対する普遍近似定理として以下に述べる．
テンソル代数
$T^r(\mathbb{R}^d)^*$や
演算$\left<\bullet,\bullet \right>$の定義は，付録\ref{tensor}の
式(\ref{tensor_algebra}), (\ref{pairing})を参照のこと．
また，経路は樹状同値類として解釈する（式(\ref{unique})の上の説明参照）．
        \begin{theorem}[\cite{Litterer2011OnAC,levin2013learning,JMLR:v20:16-314}]\label{th3}
          有界変動を持つ経路の集合$BV([0,1],\mathbb{R}^d)$のコンパクト部分集合$K$上の連続関数$f\in C(K,\mathbb{R})$と$\epsilon>0$に対して，ある$0$以上の整数$r$と
          $w \in T^r(\mathbb{R}^d)^*$があって，
  \begin{align}
    \label{approx}
    \sup_{X\in K} \left| f(X)-\left< w,\mathcal{S}(X) \right>\right|&<\epsilon
    \end{align}
とすることができる．
\end{theorem}
\begin{proof}
$K$に属する経路にシグネチャの線形結合を割り当てる関数の集合
\begin{align}
  A&=\left\{f:  K\ni X\mapsto \left< w,\mathcal{S}(X)\right> \in \mathbb{R}
  \middle|\exists r\geq 0,~ w \in T^r(\mathbb{R}^d)^*\right\}
  \end{align}
は，$K$の各点を分離する多元環となっている．すなわち，
以下の性質を持つ．
\begin{enumerate}
\item $r$段シグネチャへの変換$K \ni X\mapsto \mathcal{S}_r(X)$は連続関数なので
   \citep[Prop.\,7.15 of][]{frizvictoir2010}，$A\subset  C(K,\mathbb{R})$.
\item 関数の和とスカラー倍に関して閉じている: $f,g\in A, \lambda,\mu \in \mathbb{R} \implies
  \lambda f + \mu g \in A.$
\item シャッフル積の性質(\ref{shuffle_fg})
  により，$A$は関数の積に関して閉じている: $f,g\in A \implies f g \in A.$
  
\item $A$は定数関数$1$を含む．
\item シグネチャの一意性(\ref{unique})より，$A$は$K$の点を分離する．
  実際，$X\nsim Y$ならばある多重添字$I$に対するシグネチャの成分が異なる.
  このとき，
  関数$X \mapsto \mathcal{S}^{(I)}(X)$は$A$に属し$X$と$Y$とを分離する．
  \end{enumerate}
従って，ストーン＝ワイエルシュトラスの定理\citep{stone1937applications}より，
$A$は$C(K,\mathbb{R})$内で稠密である：$\overline{A}=C(K,\mathbb{R})$．
\end{proof}

定理\ref{th3}は，
経路の連続関数が$f(X)\approx \sum_I w_I \mathcal{S}^{(I)}(X)$というように
多重添字に亘る反復積分の線形結合によって任意の精度で近似できることを示しているので，
シグネチャの各項$\{\mathcal{S}^{(I)}(X)\}_{I\in \text{(multiindex)}}$が
基底関数の役割を果たしていることがわかる．
これは，\ref{polynomial}項で述べた多項式近似において，
単項式$\{1,x,x^2,\cdots\}$が１変数関数に対して果たす役割と同等である．
  
\section{シグネチャを用いた機械学習}
シグネチャを機械学習に応用する方法について述べ，実データへの適用例を示す．
\subsection{シグネチャの計算}
経路$X:[0,1]\to \mathbb{R}^d$の定義域
$[0,1]$の分割$0=t_0,t_1,\cdots,t_m=1$を採ると，
折れ線（区分的に線形な経路）は，各区分において
\begin{align}
(X_j)_t &= X_{t_{j-1}}+\frac{t-t_{j-1}}{t_{j}-t_{j-1}}\left(X_{t_{j}}-X_{t_{j-1}}\right),
\quad t \in [t_{j-1},t_{j}]
\end{align}
と定義される．
シグネチャを計算するにあたっては，
折れ線に対するシグネチャを考えるだけで充分であることが，
Chow–-Rashevskiiの定理 \citep[Theorem\,7.28 of][]{frizvictoir2010}
により保証されている．
また，実際的な面でも，系列データは
観測点での値を線形に結んだ折れ線とみなすことができることが多い．

折れ線に対する$n$段シグネチャは，次のように計算することができる．
ベクトル$v=(v^{(1)},\cdots,v^{(d)})\in\mathbb{R}^d$に対して，
線分$X: [0,1]\ni t \mapsto vt \in \mathbb{R}^d$を考える．
$n$次の反復積分を式(\ref{sig_n})に従って計算すると，
\begin{align}
  \mathcal{S}^{(i_1\cdots i_n)}(X)
  &=
  \frac1{n!} v^{(i_1)}v^{(i_2)}\cdots v^{(i_n)}.
\end{align}
そして，$m$個の線分$X_1,\cdots,X_m$ 
をつなげた
折れ線$X_1*\cdots*X_m$ 
に対する反復積分は，
Chenの恒等式(\ref{chen})より，$\mu=2,\cdots,m$に対して
\begin{align}
\label{chen2}
\mathcal{S}_n(X_1*\cdots*X_{\mu})&=
\mathcal{S}_n(X_1*\cdots*X_{\mu-1})\otimes\mathcal{S}_n(X_{\mu})
\end{align}
と順次計算することができる．
なお，$n$次の反復積分の計算にはそれより低次の反復積分の情報を用いるので，
各$\mu$において，$1,\cdots,n$次の反復積分に対して(\ref{chen0})を順番に計算していく．
シグネチャの計算を行うためのPythonライブラリとして，$\mathtt{esig}$\citep{esig}，
$\mathtt{iisignature}$\citep{WOS:000582337500008}などがある．
\subsection{打ち切り誤差\label{error_estimate}}
シグネチャの段数を打ち切った時の非可換テイラー展開（\ref{taylor}項）
の近似精度を調べる．

まず，$k$次の反復積分のノルム（付録\ref{tensor}の式(\ref{norm1})）
は以下を満たす\citep[Prop.\,2.2 of][]{lyons2007differential}．
$\ell$を経路$X:[0,t]\to \mathbb{R}^d$の長さとして，
\begin{align}
  \label{bound}
  \left|\int_{0<u_1<\cdots<u_k<t} dX_{u_1}\otimes   \cdots  \otimes dX_{u_k}
  \right|_{\otimes k}
    &=
  \left|
  \int_{0<u_1<\cdots<u_k<t}
  \dot{X}_{u_1}\otimes   \cdots  \otimes \dot{X}_{u_k}
  du_1\cdots du_k\right|_{\otimes k}\\
    &\leq
    \int_{0<u_1<\cdots<u_k<t} |\dot{X}_{u_1}| \cdots  
    |\dot{X}_{u_k}|du_1\cdots du_k\nonumber\\
    &=
    \int_{0<u_1<\cdots<u_k<t} v^k du_1\cdots du_k\nonumber\\
    &=\frac{\ell^k}{k!}.\nonumber
  \end{align}
2行目において，付録\ref{tensor}の式(\ref{prop_norm1})を用いた．また，
3行目において，経路$X$はほとんど至るところ微分可能で一定速度$v=\ell/t$を持つとした．
これは時間パラメータの付け替えにより常に可能である
\citep[Prop.\,2.2 of][]{lyons2007differential}．

    
評価値$y\in \mathbb{R}$がある滑らかな関数$F$を用いて$y=F(Y_t)$と表され，さらに
$Y: [0,t]\to \mathbb{R}^e$が$X$に駆動される微分方程式(\ref{ode1})を満たす場合には，
\ref{taylor}項で述べたように，形式的には$y$を
テイラー展開(\ref{taylor22})で表すことができる．
このとき，$n$項までのテイラー展開の剰余は式(\ref{taylor22a})で与えられるが，
この大きさを以下に評価する．
$\mathbb{R}^d$から$\mathbb{R}^e$への
線形変換（行列）の空間を$L(\mathbb{R}^d,\mathbb{R}^e)$と
書くことにし，
滑らかな関数
$F:\mathbb{R}^e\to \mathbb{R}$と
滑らかな変換場
$$V :\mathbb{R}^e\to L(\mathbb{R}^d,\mathbb{R}^e),\quad
y\mapsto
\begin{bmatrix}
V_{1}^{(1)}(y)&\cdots &V_{d}^{(1)}(y)\\
\vdots    &       &\vdots    \\
V_{1}^{(e)}(y)&\cdots &,V_{d}^{(e)(y)}
\end{bmatrix}$$
  に対して，
  演算子$\nabla_V^{\circ n}F: \mathbb{R}^e \to
  L((\mathbb{R}^d)^{\otimes n},\mathbb{R})$を次のように帰納的に定義する．
  \begin{align}
    \nabla_V^{\circ 0}F  &= F,\quad
    \nabla_V^{\circ (k+1)}F = \sum_{j=1}^{e}V^{(j)}\partial_j 
    (\nabla_V^{\circ k}F).
    \end{align}
また，演算子のノルムを次のように定義する．
  \begin{align}
    \left\|    \nabla_V^{\circ n}F  \right\|_{\infty}
    &=\sup_{y\in \mathbb{R}^e} \left| \nabla_{V}^{\circ n}F(y)  \right|,\quad
    \left| \nabla_{V}^{\circ n}F(y)  \right|
    =
    \sup_{x\in (\mathbb{R}^d)^{\otimes n}}\frac{\left|
      \nabla_{V}^{\circ n}F(y)x \right|}{|x|_{\otimes n}}.
    \end{align}
  このとき，
  剰余項(\ref{taylor22a})に対して次の評価が成り立つ
  (\citet{10.1214/ECP.v20-4124}の式(1.5)参照)．
\begin{align}
  |R_{n+1}(t)|
  &=
  \left|
  \int_{0<u_1<\cdots<u_n<t}
  \left(  \nabla_{V}^{\circ n} F(Y_{u_1})-
  \nabla_{V}^{\circ n} F(Y_{0})\right)
    dX_{u_1}\otimes \cdots  \otimes dX_{u_n}
    \right|
    \label{error}\\
  &=
  \left|
  \int_{0<u_0<\cdots<u_n<t}
    \nabla_{V}^{\circ (n+1)} F(Y_{u_0})
    \dot{X}_{u_0}\otimes    \cdots  \otimes \dot{X}_{u_n}
    du_0\cdots du_n
    \right|\nonumber\\
  &\leq
  \int_{0<u_0<\cdots<u_n<t}
  \left|
    \nabla_{V}^{\circ (n+1)} F(Y_{u_0})
    \dot{X}_{u_0}\otimes    \cdots  \otimes \dot{X}_{u_n}
    \right|
    du_0\cdots du_n 
    \nonumber\\
  &\leq
  \int_{0<u_0<\cdots<u_n<t}
  \left|    \nabla_{V}^{\circ (n+1)} F(Y_{u_0})   \right|
    v^{n+1} du_0\cdots du_n 
    \nonumber\\
    &\leq
    \left\|    \nabla_{V}^{\circ (n+1)} F   \right\|_{\infty} v^{n+1}
    \int_{0<u_0<\cdots<u_n<t}
    du_0\cdots du_n 
    \nonumber\\
   &=
    \frac{\ell^{n+1}}{(n+1)!}
    \left\|    \nabla_{V}^{\circ (n+1)} F \right\|_{\infty}.\nonumber
\end{align}
4行目において，式(\ref{bound})と同様に経路$X$はほとんど
至るところ微分可能で一定速度$v=\ell/t$を持つとした．

  まとめると，
  滑らかなベクトル場の集合$V$と滑らかな関数$F$が
    $\|    \nabla_{V}^{\circ n} F    \|_{\infty}=o(n!/\ell^n)$
  を満たす時，テイラー展開は収束し，$n$項までの展開の
  誤差は (\ref{error})で抑えられる．

\subsection{線型回帰}
\ref{approx_theory}項で述べたように，シグネチャを用いて経路に対する非線形関数を近似することができる．
このことを利用して，実際の系列データに対して機械学習を適用することを考える．
普遍近似定理によれば，$d$次元経路$X$の集合$K$に対して定義される
実数値連続関数$f$は，十分大きな段数$r$のシグネチャを
線型変換したものと実際上みなしてよい．従って，
実データ$y \in \mathbb{R}$は，これに独立同分布に従うノイズを加えた
以下のシステムから生成されているとする：
\begin{align}
  \label{truesystem}
  y &= \left<w_{\text{true}},\mathcal{S}_r(X)\right>+\zeta,\\
  \mathbf{E}[\zeta|X] &=0,\quad \mathbf{E}[\zeta^2|X]=\sigma^2,\quad
  w_{\text{true}} \in T^r(\mathbb{R}^d)^*.
  \nonumber 
\end{align}
ここで，$\mathbf{E}[\zeta^2|X]=\sigma^2$は重回帰式のの残差が
有限の分散を持つことを表している．
一方，このシステムに対するモデルは，段数を$n<r$に減らし，
ノイズレベルを$\Sigma \geq \sigma$に上げて，
\begin{align}
  \label{model}
  y &= \left<w,\mathcal{S}_n(X)\right>+\xi,\\
  \mathbf{E}[\xi|X] &=0,\quad \mathbf{E}[\xi^2|X]=\Sigma^2,\quad
  w \in T^n(\mathbb{R}^d)^*
  \nonumber
\end{align}
とする．ここで，重み$w$を列ベクトル
$\mathbf{w}\in \mathbb{R}^{M_n\times 1}, ~M_n=(d^{n+1}-1)/(d-1)$
とみなす．
また，訓練データセットとして，系列データ$X_i$のシグネチャと評価値$y_i$との組が
\begin{align}
  \mathcal{D} &=\left\{(\mathcal{S}_n(X_i),y_i)|  
  ~i=1,\cdots,N   \right\}, \quad
\end{align}
と与えられているとして，計画行列
$\mathbf{X}\in \mathbb{R}^{N \times M_n}$と
列ベクトル$\mathbf{y}\in \mathbb{R}^{N \times 1}$の成分を
次のように定義する．
\begin{align}
  (\mathbf{X})_{ik} &= \mathcal{S}_r^{(k)}(X_i),\quad (\mathbf{y})_i=y_i,
  \quad 1\leq i \leq N,~1\leq k \leq M_n.
\end{align}
ここでの付番$(k)$は，多重添字の通し番号とする．
ガウス＝マルコフの定理\citep[e.g.,][]{10.1093/biomet/36.3-4.458}より，
このデータ$\mathcal{D}$のもとで最適な重みは，コスト関数：
\begin{align}
  \label{cost}
  J(\mathbf{w}) &= 
  \left( \mathbf{X} \mathbf{w} -\mathbf{y} \right)^T
  \left( \mathbf{X} \mathbf{w} -\mathbf{y} \right)
\end{align}
を最小化することによって得られる．特に，$N>M_n$で，
列ベクトル$x^{(1)},\cdots,x^{(M_n)}$が線型独立のとき，
この訓練データセット$\mathcal{D}$に基づく最適な重み$\mathbf{w}_{\mathcal{D}}$は，
\begin{align}
  \mathbf{w}_{\mathcal{D}} &= \left(\mathbf{X}^T \mathbf{X}\right)^{-1}
  \mathbf{X}^T\mathbf{y}
\end{align}
となるが，この重みを用いた$\mathbf{y}$の予測値は，
\begin{align}
  \mathbf{X}\mathbf{w}_{\mathcal{D}}&=\mathbf{P}\mathbf{y},\quad
  \mathbf{P} :=
  \mathbf{X}\left(\mathbf{X}^T \mathbf{X}\right)^{-1}  \mathbf{X}^T
\end{align}
と書ける．ここで，$\mathbf{P}\in \mathbb{R}^{N\times N}$
は射影行列または影響行列と呼ばれ，
\begin{align}
  \label{influence}
  \mathbf{P}^2 &=\mathbf{P},\quad
  \mathbf{P}^T =\mathbf{P},\quad
  \mathrm{tr}(\mathbf{P})=M_n
\end{align}
などの性質を持つ\citep[e.g.,][]{cardinali2004influence}．
ここで，$\mathrm{tr}$は対角成分の和．

システム(\ref{truesystem})が生成したデータ$y$を
をモデル(\ref{model})で予測する際の二乗誤差は，
ノイズを変化させたときの期待値を$\mathbf{E}_{\zeta}$，
観測データセットを変化させたときの期待値を$\mathbf{E}_{\mathcal{D}}$，
系列データに亘る期待値を$\mathbf{E}_X$として，
\begin{align}
  \label{bias-variance}
&  \mathbf{E}_{\zeta,\mathcal{D},X}  
  \left[ \left(  y-\left<w_{\mathcal{D}},\mathcal{S}_n(X)\right>  \right)^2\right]
  \\
  &=
  \mathbf{E}_{\zeta}\left[\left(y-\left<w_{\text{true}},\mathcal{S}_r(X)\right>\right)^2\right]
  +
  \mathbf{E}_{X}\left[
    \left(\left<w_{\text{true}},\mathcal{S}_r(X)\right>
    -\left<\bar{w},\mathcal{S}_n(X)\right>\right)^2\right]
  +
  \mathbf{E}_{\mathcal{D},X} \left[
    \left<w_{\mathcal{D}}-\bar{w}    ,\mathcal{S}_n(X)\right>^2    \right]\nonumber\\
  &\leq
  \mathbf{E}_{\zeta}\left[\zeta^2   \right]
  +
  \mathbf{E}_{X}\left[\left<\pi_{>n}(w_{\text{true}}),\mathcal{S}_r(X)\right>^2\right]
  +
  \mathbf{E}_{\mathcal{D}} \left[\frac1N \Xi^T\mathbf{P} \Xi   \right]
  \nonumber \\
  &\leq  \sigma^2+  
  \left(
\frac{\ell^{n+1}}{(n+1)!}\left\|\nabla^{\circ (n+1)}_V F \right\|_{\infty}
   \right)^2
  +\frac{M_n}{N}\Sigma^2 \nonumber
\end{align}
を満たす．ここで，
$\bar{w}=  \mathbf{E}_{\mathcal{D}}\left[w_{\mathcal{D}}\right]$，
$\Xi$は平均$0$で分散$\Sigma^2$の
独立同分布確率変数$N$個からなる列ベクトル，
$\ell$は$K$内の経路の長さの上限．
なお，右辺各項の評価においては，
式(\ref{truesystem})，(\ref{error})，および(\ref{influence})をそれぞれ用いた．
右辺の各項は，順に雑音（ノイズ），偏り（バイアス），分散（バリアンス）と呼ばれる
\citep[e.g.,][]{bishop2006,hastie_09_elements-of.statistical-learning,Mehta2019}．

サンプル数$N$が少ないときに
シグネチャの段数$n$を増やすと，$\Sigma$はやや減るものの$M_n$が大きく増え，
分散が大きくなる場合がある（過剰適合）．一方，$n$を減らすと偏りが大きくなる（過少適合）．
このようなバランスを考慮して，２乗誤差を小さくする$n$を決める必要がある．
これが，この問題に対する偏りと分散のトレードオフである．
過剰適合を回避するには，シグネチャの段数を減らす他にも，
後述のように，
コスト関数(\ref{cost})に$L_1$罰則項をつけることによって重要度の低い重み成分を
除外する方法がある\citep{10.2307/2346178}．

\subsection{実データへの適用例}
全球海洋観測プロファイルArgo\citep{doi:10.1029/2004EO190002}の品質管理に適用した例が，\citet{Sugiura2020}に示されている．

Argoフロートは，約$2000\unit{m}$の海中から海面まで浮上してゆき，
その間に圧力(P)，塩分(S)，水温(T)を観測する．
得られた観測プロファイルは，$(P,S,T)$座標を持つ
$\mathbb{R}^3$内の経路とみなすことができる．
ここでは，機械学習の精度を向上させるため，遅れ座標を導入しLead-lag変換
  \citep[e.g.,][]{chevyrev2016primer,FERMANIAN2021107148}という処理を施して，
  $\mathbb{R}^6$内の経路とみなす$(d=6)$．
各経路$X$には，評価値$y=0$または$1$が付与されているものとする．
$y=1$は品質管理を合格したという意味であり，$y=0$は不合格という意味である．
経路は，$n=6$段のシグネチャに変換してから用いる．
データセットの任意の4割を訓練データセットとして，残りの6割を検証データセットとする．
そして回帰における重みを訓練データセットから計算し，その重みを使って検証データセットに属する
プロファイル$X$に対する$y$を予測できるかどうかを調べる．

まず，訓練データセット$\{(\mathcal{S}_n(X_i),y_i)|i=1,\cdots,N\}$に対して，
式(\ref{cost})に$L_1$罰則項を加えた次のコスト関数を最小化する重み$w_*$を求める．
\begin{align}
  \label{cost2}
  J(w) &=
  \sum_{i=1}^N \left( \left<w, \mathcal{S}_n(X_i)\right>  -y_i \right)^2
  +
  \alpha \left\|w\right\|_1.
\end{align}
ここで，$\|\bullet \|_1$は多重添字に亘る絶対値の和を表す．

次に，検証データ$X$に対する予測値$\widetilde{y}$は，次式で求められる．
\begin{align}
  \label{pred2}
  \widetilde{y} &= \left<w_*, \mathcal{S}_n(X)\right>.
\end{align}
なお，式(\ref{cost2})内の$\alpha>0$は，検証データに対する予測(\ref{pred2})の成績がよくなるように調整する．
これは，式(\ref{bias-variance})の誤差を小さくするような$w$の自由度を探していることに相当する．
実際に採用されたのは，サンプル数$N=3.2\times 10^4$に対して自由度$6.7\times 10^3$である．

Argoプロファイルと対応するシグネチャの例を図\ref{fig:argo}に示す．
また，図\ref{fig:hist}は訓練データ(左)と検証データ(右)に対する
予測結果をヒストグラムで表したものである．
少なくとも，しきい値$0.5$よりも低い推定値のデータに，
合格プロファイルが混入することはほとんどないとは言える．
一方，同じしきい値で不合格プロファイルの数割は検知することができることもわかる．
論文内では，シグネチャ法を用いた方が対照ケースに比して成績が良いことが示されている．
\section{まとめ}
シグネチャは経路の連続関数の空間の基底関数であり，
系列データとその評価値との間の非線形な関係を表すのに極めて有効である．
このことを使って，観測された系列データに対する機械学習を行うことができる．
典型的な適用例として，従来の多項式回帰をシグネチャを基底関数とする
回帰に置き換えて，系列データに対する教師あり学習を行うことができる．
この手順を海洋観測プロファイルとその品質管理フラグに適用した例を示した．

シグネチャ法に関しては，この他にもいろいろな応用の可能性が考えられる．
\begin{itemize}
\item   一般に時系列解析は，多次元の経路を読み取ることに帰着するので，
  時系列の断片を経路と見てシグネチャに変換することで，同様の回帰を行うことができる．
  特に，将来における何らかの値を経路に対する評価値とみなせば，
  将来予測が可能である\citep[e.g.,][]{sugiura2021simple}．
  この手法においては，非線形状態空間モデルを経路自体が持っている非線形性（シグネチャ）
  と線型状態空間モデル（重み）とに分離するため，
  回帰が極めて簡単になる．
\item  また，データ同化においても，観測されたプロファイルとモデル内のプロファイルを
  比較する必要が生じるが，そのような場合に
  両者のシグネチャの差を縮めるというようなコスト関数を設定することにより，
  品質の向上が期待できる．
  なぜなら，経路上の各点が近いという線型の問題設定から，
  経路の性質が近いという非線形の問題設定へと自然に転換することができるからである．
\end{itemize}

\section*{謝辞}
貴重な指摘をいただいた査読者に感謝の意を表する．
本研究は，日独仏AI研究の研究課題「強化型データストリーム解析：ラフパス理論と機械学習アルゴリズムの融合」（JST-PROJECT-20218919, 期間：2020年12月--2024年3月）の一環として行われた．
\appendix
\section{テンソル表記\label{tensor}}
以下のようにテンソルの概念を導入すると，
シグネチャをテンソル代数の元として表すことができる．

\subsection{テンソル代数\label{app-a1}}
実ベクトル$u,v \in \mathbb{R}^d$に対して，テンソル積$u \otimes v$は，双対ベクトル空間の元$f,g$に対して，
\begin{align}
  \label{dual}
  (u \otimes v)(f,g) \mapsto \left< u,f \right> \left< v,g \right> &\in \mathbb{R}
\end{align}
なる実数を与える双線形写像である．ここで，$\left<\bullet,\bullet\right>$はduality pairingである．
座標表示すると，$u=(u_1,\cdots,u_d),~v=(v_1,\cdots,v_d),~f=(f^1,\cdots,f^d)$
に対して，$\left<u,f\right>=\sum_i u_i f^i$であり，
$u\otimes v \in \mathbb{R}^{d\times d}$はランク１行列
$(u_i v_j)_{i,j=1,\cdots,d}$を成分として持つ．
また，テンソル積$\otimes$の演算自体も双線形である．すなわち，
\begin{align}
  (a u + b v)\otimes  w &=
  a (u \otimes  w)+  b (v \otimes  w),\\
w\otimes   (a u + b v) &=
  a (w \otimes  u)+  b (w \otimes  v).
\end{align}
ここで，$u,v,w \in \mathbb{R}^d,~ a,b \in \mathbb{R}$．
なお，テンソル積$u\otimes v$と$v\otimes u$とは一般に等しくない（可換でない）．

ベクトル空間$\mathbb{R}^d$の基底を$\{e_1,\cdots,e_d\}$とし，テンソル積で$n$個のベクトルのテンソル積の線形結合の全体($n$階テンソルともいう)を
$(\mathbb{R}^d)^{\otimes n}$とかく．すなわち，
\begin{align}
  (\mathbb{R}^d)^{\otimes n}&=\left\{
  \sum_{i_1,\cdots, i_n=1,\cdots,d} g^{(i_1\cdots i_n)}
  e_{i_1}\otimes \cdots \otimes e_{i_n} \middle|
  g^{(i_1\cdots i_n)}\in \mathbb{R}
  \right\}
\end{align}
さらに，$k=0,1,\cdots,n$階のテンソルの直和
\begin{align}
  T^n(\mathbb{R}^d):=\bigoplus_{k=0}^n(\mathbb{R}^d)^{\otimes k}
  &=\left\{
  \sum_{k=0}^n
  \sum_{i_1,\cdots, i_{k}=1,\cdots,d}
  g^{(i_1\cdots i_{k})}
  e_{i_1}\otimes \cdots \otimes e_{i_{k}}
  \middle|
  g^{(i_1\cdots i_{k})}\in \mathbb{R}
  \right\}\label{tensor_algebra}
\end{align}
をテンソル代数という．

\subsection{テンソル代数における積\label{app-a2}}
テンソル代数の元には
和と非可換な乗算が定義されるので，非可換多項式からなる代数とみることができる．
すなわち，ベクトルに対するテンソル積を自然に拡張して，積を
\begin{align}
  \left(
  g^{(i_1\cdots i_{k})} e_{i_1} \otimes \cdots \otimes e_{i_{k}}
  \right)
  \otimes
  \left(
  h^{(j_1\cdots j_{l})}
  e_{j_1} \otimes \cdots \otimes e_{j_{l}}
  \right)
  &=
  g^{(i_1\cdots i_{k})}
  h^{(j_1\cdots j_{l})}
  e_{i_1} \otimes \cdots \otimes e_{i_{k}}\otimes
  e_{j_1} \otimes \cdots \otimes e_{j_{l}}
\end{align}
と定義する．
$k$階テンソルと$l$階テンソルとのテンソル積は$k+l$階テンソルになる．
このテンソル積はテンソル代数の元に双線型に拡張される．例えば，
$g,h\in T^2(\mathbb{R}^d)$のテンソル積は，$T^2(\mathbb{R}^d)$内では
以下のように算定される．
\begin{align}
  \label{tensor_prod}
  &\left(g^{()}+
  \sum_{i_1=1,d} g^{(i_1)}e_{i_1}+
  \sum_{i_1,i_2=1,d} g^{(i_1i_2)}e_{i_1}\otimes e_{i_2}\right)
  \otimes 
  \left(h^{()}+
  \sum_{i_1=1,d} h^{(i_1)}e_{i_1}+
  \sum_{i_1,i_2=1,d} h^{(i_1i_2)}e_{i_1}\otimes e_{i_2}\right)\\
  &=
  g^{()}h^{()}+
  \sum_{i_1=1,d} (g^{()}h^{(i_1)}+g^{(i_1)}h^{()})e_{i_1}+
  \sum_{i_1,i_2=1,d} (g^{()}h^{(i_1i_2)}+g^{(i_1)}h^{(i_2)}+g^{(i_1i_2)}h^{()})
  e_{i_1}\otimes e_{i_2}.
  \nonumber
\end{align}
$3$階以上のテンソルは切り捨てられていることに注意．

\subsection{シグネチャのテンソル表記\label{app-a4}}
以上の定義より，反復積分を次のように簡潔に表記できる．
\begin{align}
  \sum_{i_1\cdots i_n}\mathcal{S}^{(i_1\cdots i_n)}(X) e_{i_1\cdots i_n}
  &=
  \sum_{i_1\cdots i_n}
  \int_0^1  \int_0^{t_n} \cdots  \int_0^{t_2}
  dX^{(i_1)}_{t_1}    \cdots  dX^{(i_{n-1})}_{t_{n-1}}dX^{(i_n)}_{t_n}
  e_{i_1\cdots i_{n-1} i_n}\\
  &=\sum_{i_1\cdots i_n}
  \int_0^1  \int_0^{t_n} \cdots  \int_0^{t_2}
  dX^{(i_1)}_{t_1}e_{i_1}    \cdots  dX^{(i_{n-1})}_{t_{n-1}}e_{i_{n-1}}
  dX^{(i_n)}_{t_n}e_{i_{n}}\nonumber\\
  &=
  \int_{0<t_1<\cdots<t_n<1}
  dX_{t_1}\otimes   \cdots  \otimes dX_{t_n}.\nonumber
\end{align}
ここで，$(\mathbb{R}^d)^{\otimes n}$の基底を
簡単のため$e_{i_1\cdots i_n}:=e_{i_1}\otimes \cdots \otimes e_{i_n}$と書いた．
シグネチャは反復積分の直和なので，次のように書ける．
\begin{align}
  \mathcal{S}_n(X) &=
  \sum_{k=0}^{n}
  \sum_{i_1\cdots i_n}\mathcal{S}^{(i_1\cdots i_k)}(X) e_{i_1\cdots i_k}
    \label{sig_tensor}\\
    &=\sum_{k=0}^{n}
  \int_{0<t_1<\cdots<t_k<1}
  dX_{t_1}\otimes   \cdots  \otimes dX_{t_k}
  \in T^n(\mathbb{R}^d).\nonumber
\end{align}

\subsection{テンソル代数のノルム\label{app-a3}}
$g=\sum_{i_1,\cdots, i_k=1,\cdots,d} g^{(i_1\cdots i_k)}
  e_{i_1}\otimes \cdots \otimes e_{i_k} \in(\mathbb{R}^d)^{\otimes k}$に対して，
\begin{align}
  |g|_{\otimes k}&=\sqrt{\sum_{i_1,\cdots,i_k=1,\cdots,d}\left(g^{(i_1\cdots i_k)}\right)^2}
  \label{norm1}
\end{align}
とノルムを定義すると，
$g\in (\mathbb{R}^d)^{\otimes k},~
h\in (\mathbb{R}^d)^{\otimes l}$に対して，
\begin{align}
|g\otimes h|_{\otimes k+l}
&=
|g|_{\otimes k}|h|_{\otimes l}
  \label{prop_norm1}
\end{align}
が成り立つ\citep[chapter\,7 of][]{frizvictoir2010}．
さらに，ノルム(\ref{norm1})をもとに，
$g=g^0+\cdots+g^n\in T^n(\mathbb{R}^d),~g^k\in (\mathbb{R}^d)^{\otimes k}$
のノルムを
\begin{align}
\|g\| &= \max_{k=0,\cdots,n}|g^k|_{\otimes k}
  \label{norm2}
\end{align}
と定義することができる\citep[chapter\,7 of][]{frizvictoir2010}．

\subsection{シグネチャの線形汎関数\label{app-a5}}
$T^n(\mathbb{R}^d)$の双対空間$T^n(\mathbb{R}^d)^*$
の基底を$e^{i_1\cdots i_k}$と書き，
シグネチャの線型汎関数
\begin{align}
  w &=
  \sum_{k=0}^{n}
    \sum_{i_1\cdots i_k}w_{i_1\cdots i_k}(X) e^{i_1\cdots i_k}
    \in T^n(\mathbb{R}^d)^*
\end{align}
をとる．
これを，演算規則$
\left<e_{i_1i_2},e^{j_1j_2}\right>=\delta_{i_1}^{j_1}\delta_{i_2}^{j_2}$
などに従って，
シグネチャ(\ref{sig_tensor})に双線形に作用させると，
\begin{align}
\left<w,\mathcal{S}_n(X)\right>
&=
  \sum_{k=0}^{n}
  \sum_{i_1\cdots i_k}w_{i_1\cdots i_k}\mathcal{S}^{(i_1\cdots i_k)}(X)\in \mathbb{R}
  \label{pairing}
\end{align}
というスカラーが得られる．
これは経路の関数$X\mapsto
\left<w,\mathcal{S}_n(X)\right>$を定義している．

\begin{figure}
 \begin{center}
   \includegraphics[width=20em]{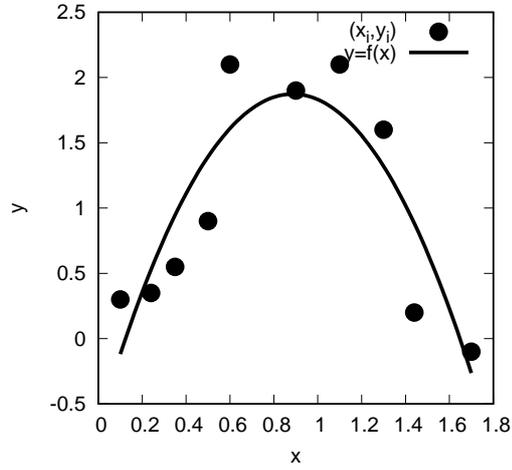}
  \caption{\label{fig:reg}多項式回帰$(y=w_0+w_1 x + w_2 x^2)$}
 \end{center}
\end{figure}

\begin{figure}
 \begin{center}
   \includegraphics[width=50em]{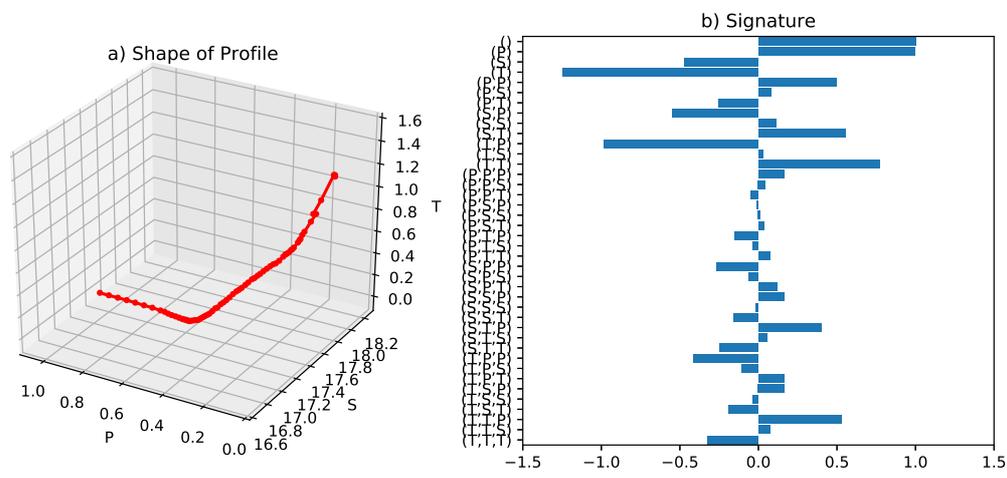}
   \caption{\label{fig:argo}Argoプロファイル(a)とそのシグネチャ(b)．
     圧力(P)・塩分(S)・水温(T)を規格化して，３次元空間上にプロファイル$X$を表示．
     シグネチャは，例えば多重添字$(PTS)$に対して$\mathcal{S}^{(PTS)}(X)$
     の値を表示．}
 \end{center}
\end{figure}
\begin{figure}
 \begin{center}
   \includegraphics[width=22em]{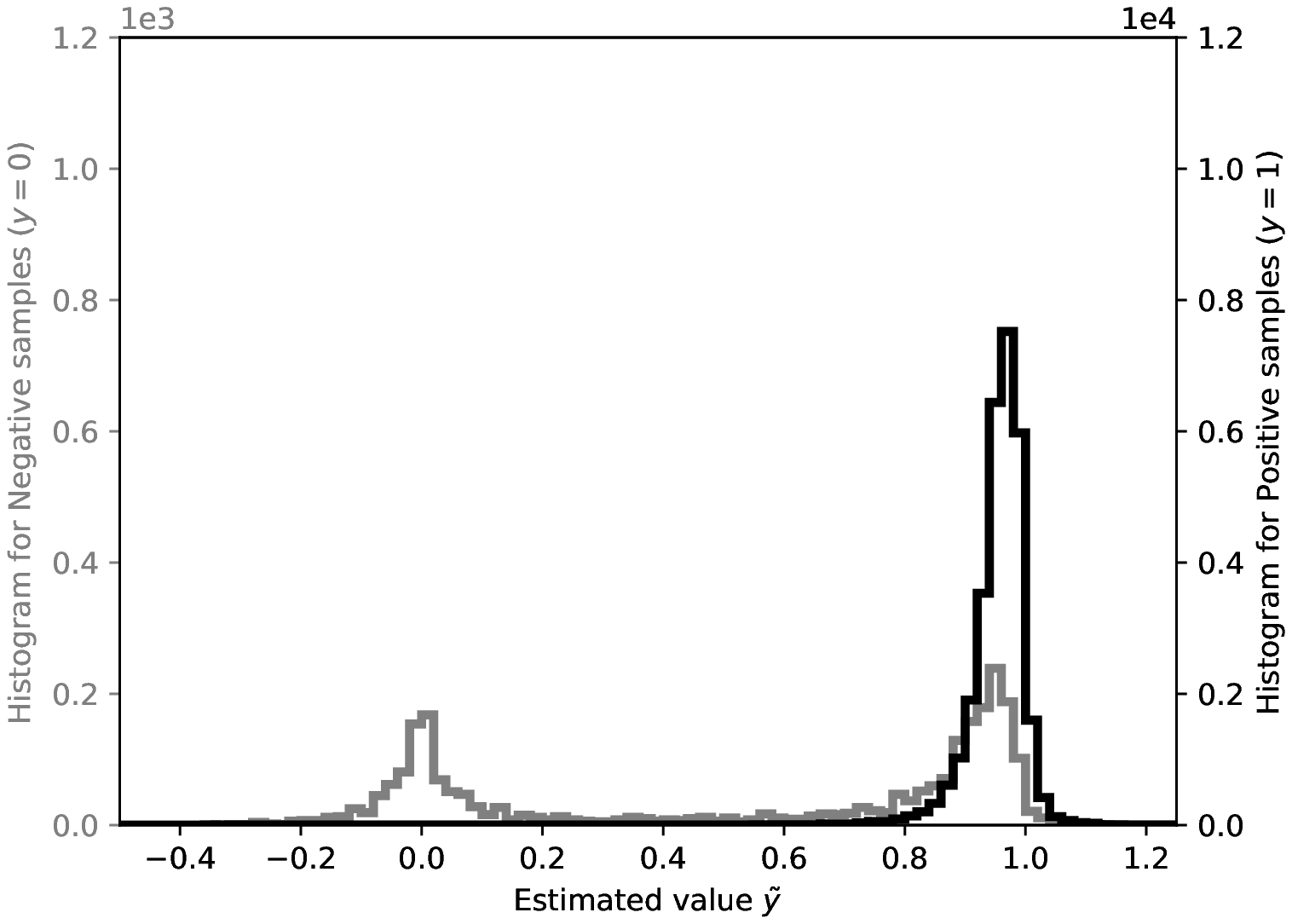}
   \includegraphics[width=22em]{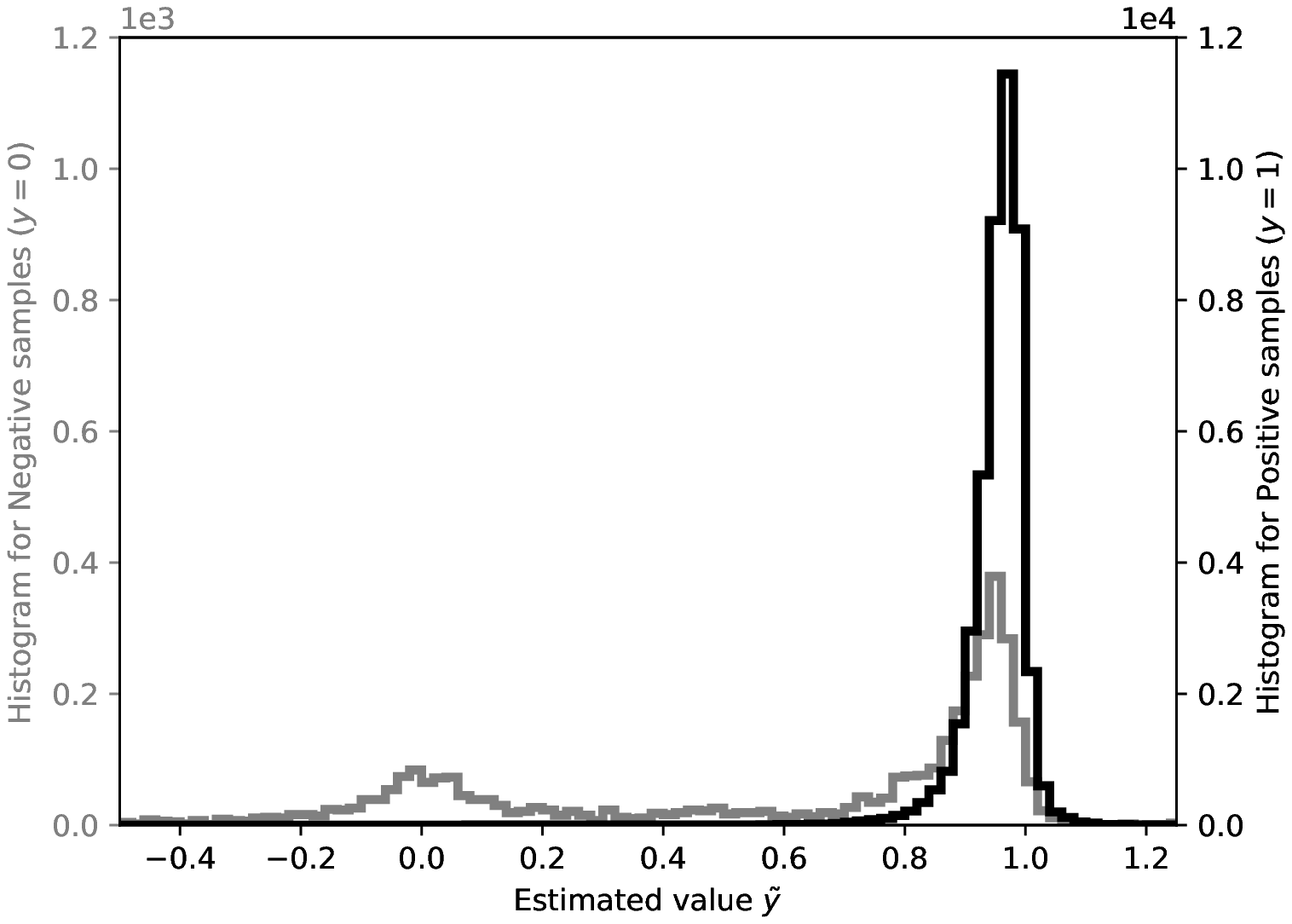}
   \caption{\label{fig:hist}訓練データ（左）と検証データ（右）に対する予測値（横軸）
     のビン毎の度数（縦軸）．
     黒は$y=1$，灰色は$y=0$のデータに対する予測値の度数分布を表し，
     前者の度数目盛りを右軸に，後者を左軸に示す．}
\end{center}
\end{figure}

\bibliographystyle{pism}
\bibliography{myrefs}
\end{CJK}
\end{document}